\documentclass[a4paper,11pt]{article}

\usepackage{lineno,hyperref}
\usepackage{amsmath,amssymb,amsthm}
\usepackage{authblk}

\usepackage{mathtools}
\usepackage{amsopn}
\usepackage{tabularx,lipsum,environ}
\usepackage{paralist}
\usepackage{microtype}
\usepackage{authblk}
\usepackage{hyperref}
\usepackage{fancyhdr}
\usepackage{rotating}
\usepackage{overpic}
\usepackage{ucs}
\usepackage{enumerate}
\usepackage{graphicx}
\usepackage{booktabs}
\usepackage{varwidth}
\usepackage{subfig}
\usepackage{todonotes}

\usepackage{tikz}
\usepackage{ascmac}

\newcommand{\polylog}{{\rm polylog}}

\newcommand{\mw}{{\rm mw}}
\newcommand{\vc}{{\rm vc}}

\theoremstyle{definition}
\newtheorem{theorem}{Theorem}
\newtheorem{lemma}{Lemma}
\newtheorem{corollary}{Corollary}

\title{Finding a Maximum Minimal Separator:\\ Graph Classes and Fixed-Parameter Tractability\thanks{{This work is partially supported by JSPS KAKENHI Grant Number JP19K21537, JP20K19742, JP20H00595, JP18H05291, JP20K11692, and JST CREST JPMJCR1402.}}}

\author[1]{Tesshu Hanaka\thanks{\texttt{hanaka.91t@g.chuo-u.ac.jp}}}
\author[2]{Yasuaki Kobayashi\thanks{\texttt{kobayashi@iip.ist.i.kyoto-u.ac.jp}}}
\author[2]{Yusuke Kobayashi\thanks{\texttt{yusuke@kurims.kyoto-u.ac.jp }}}
\author[3]{Tsuyoshi Yagita\thanks{\texttt{yagita.tsuyoshi307@mail.kyutech.jp}}}

\affil[1]{Chuo University, Tokyo, Japan}
\affil[2]{Kyoto University, Kyoto, Japan}
\affil[3]{Kyushu Institute of Technology, Fukuoka, Japan}
\date{}

\begin{document}
\maketitle

\begin{abstract}
\noindent
We study the problem of finding a maximum cardinality minimal separator of a graph. This problem is known to be NP-hard even for bipartite graphs. In this paper, we strengthen this hardness by showing that for planar bipartite graphs, the problem remains NP-hard. Moreover, for co-bipartite graphs and for line graphs, the problem also remains NP-hard. On the positive side, we give an algorithm deciding whether an input graph has a minimal separator of size at least $k$ that runs in time $2^{O(k)}n^{O(1)}$. We further show that a subexponential parameterized algorithm does not exist unless the Exponential Time Hypothesis (ETH) fails. Finally, we discuss a lower bound for polynomial kernelizations of this problem.
    

\end{abstract}

\section{Introduction}
Let $G = (V, E)$ be a graph and let $a, b \in V$ be distinct vertices.
We say that $S \subseteq V \setminus \{a, b\}$ is an {\em $a,b$-separator} of $G$ if there is no path between $a$ and $b$ in the graph obtained from $G$ by deleting every vertex in $S$ and its incidental edges. 
An $a, b$-separator $S$ is {\em minimal} if there is no $a, b$-separator that is a proper subset of $S$.
A {\em minimal separator} of $G$ is a minimal $a, b$-separator for some $a, b \in V$.

Dirac \cite{Dirac:1961} introduced the notion of minimal separators to characterize the class of chordal graphs.
This notion plays an indispensable role in computing treewidth and minimum fill-in, which are deeply related to minimal chordal completions of graphs.
In particular, if the number of minimal separators in a graph is polynomially bounded in the number of vertices, the treewidth and a minimum fill-in can be computed in polynomial time~\cite{Bouchitte:2002}.
Moreover, Fomin et al.~\cite{Fomin:2015} gave a polynomial-time algorithm for finding a largest induced subgraph of bounded treewidth satisfying some logical properties $\phi$ on this classes of graphs\footnote{More precisely, their algorithm runs in time $O(f(t, |\phi|)n^{t+4} {\tt \#pmc})$. Here, $f$ is a computable function, $|\phi|$ is the length of the formula $\phi$, $t$ and $n$ are the treewidth and the number of vertices of the input graph, respectively. ${\tt \#pmc}$ denotes the number of potential maximal cliques, which is upper bounded by a polynomial in $n$ plus the number of minimal separators.}.
Given this, it is important to know how many minimal separators graphs have.
Gaspers and Mackenzie \cite{Gaspers:2018} showed that every graph with $n$ vertices has $O^*(\tau^n)$ minimal separators, where $\tau = 1.618\cdots$ is the golden ratio\footnote{The notation $O^*$ suppresses a polynomial factor.}. 
They showed that there is a graph with $\Omega(1.4422^n)$ minimal separators.
Fomin et al. improved the upper bound $O^*(\tau^n)$ to $O^*(\tau^{\mw(G)})$, where $\mw(G)$ is the modular-width of $G$ \cite{Fomin:2018}, where $\mw(G) \le n$ for every graph $G$ with $n$ vertices.
They also showed that there are $O^*(3^{\vc(G)})$ minimal separators in any graph $G$, where $\vc(G)$ is the vertex cover number of $G$.
In \cite{Fomin:2008}, they proved that AT-free graphs have still an exponential number of minimal separators but the upper bound can be improved to $O^*(2^{n/2})$.

There are several graph classes having a polynomial number of minimal separators\footnote{We say that a graph class $\mathcal G$ has a polynomial number of minimal separators if there is a polynomial $p: \mathbb N \to \mathbb N$ such that for every $G \in \mathcal G$ with $n$ vertices, $G$ has at most $p(n)$ minimal separators.}, such as circle graphs, circular-arc graphs, co-comparability graphs of bounded dimension, weakly chordal graphs \cite{Kloks:1998,Bodlaender:1995,Bouchitte:2002}.
Milani\v{c} and Piva\v{c} \cite{Milanic:2019} also studied graph classes having a polynomial number of minimal separators with respect to forbidden induced subgraphs.
If the maximum size $k$ of minimal separators is bounded, we can obtain polynomial-time algorithms for the aforementioned problems, such as computing the treewidth and a minimum fill-in of graphs, since the number of minimal separators is obviously polynomial (i.e. $O(n^{k})$).

Even if graphs have only bounded size minimal separators, they are not treewidth-bounded (e.g. block graphs with unbounded clique numbers).
However, Skodinis~\cite{Skodinis:1999} showed that many graph problems, including {\sc Maximum Independent Set}, {\sc Minimum Feedback Vertex Set}, {\sc Minimum Fill-In}, {\sc Clique}, and {\sc Graph Coloring}, can be solved in polynomial time on this class of graphs\footnote{Also, {\sc Treewidth} can be solved in polynomial time on this class of graphs, while it was not mentioned in \cite{Skodinis:1999}.}. 
More precisely, if the maximum size of a minimal separator of an $n$-vertex graph is at most $k$, these problems can be solved in $f(k)n^{O(1)}$ time, that is, they are fixed-parameter tractable parameterized by the maximum size of a minimal separator.
Note that the algorithm of Fomin et al. \cite{Fomin:2015} also runs in polynomial time on such classes of graphs to particularly solve {\sc Maximum Independent Set} and {\sc Minimum Feedback Vertex Set}, but their running time depends on the number of potential maximal cliques, which yields a running time bound $n^{O(k)}$.

Motivated by these results, Hanaka et al. \cite{Hanaka:2019} studied the problem of finding a maximum cardinality minimal separator of graphs.
Formally, the problem is defined as follows.
%\medskip\noindent

\begin{screen}[4]
\begin{tabular}{l}
     \textsc{Maximum Minimal Separator}\\
     \textbf{Input}: A graph $G$ and a non-negative integer $k$.\\
     \textbf{Goal}: Determine if $G$ has a minimal separator of size at least $k$.
\end{tabular}
\end{screen}
%\medskip

%They also studied a weighted version of this problem and discussed some application to supply chain network analysis.
Unfortunately, \textsc{Maximum Minimal Separator} is NP-hard even if the input graph is restricted to unweighted bipartite graphs~\cite{Hanaka:2019}. They also proved that the problem can be solved in polynomial time for bounded treewidth graphs.
Based on this tractable result, they claimed that there is a  $2^{p(k)}n^{O(1)}$-time algorithm for deciding whether the input graph has a minimal separator of size at least $k$, where $p(k)$ depends on the current best bound on the polynomial excluded grid theorem, which is $p(k) = O(k^{9}\polylog\ k)$ \cite{Chuzhoy:2019}.
However, the algorithm has a flaw.
Precisely speaking, they used the following approach, which is known as a win-win approach for treewidth.
By the excluded grid theorem, every graph has either small treewidth or a large grid as a minor.
If the treewidth of the input graph $G$ is at most $p(k)$, they proved that a largest minimal separator can be found efficiently by dynamic programming based on tree decompositions.
Otherwise, they claimed that $G$ has a $k \times k$ grid as a minor and then $G$ has a minimal separator of size at least $k$.
However, there are counterexamples to this claim: The complete graph of $k^2$ vertices has a $k\times k$ grid as a minor but has no minimal separators at all.
We can also construct infinitely many non-complete graphs that have huge treewidth but no large minimal separators.

\subsection{Our contribution}
In this paper,  we strengthen their hardness result and give an even faster correct algorithm for \textsc{Maximum Minimal Separator} than the previous claimed algorithm in \cite{Hanaka:2019}.

\begin{figure}
    \centering
    \begin{tikzpicture}[every node/.style={draw, rectangle,rounded corners}]\small
        \node[draw,rectangle,rounded corners,red]  (pl) at (0cm, 4cm) {planar};
        \node[draw,rectangle,rounded corners,align=center,text width=1.6cm,red] (bpl) at (0cm, 1.2cm) {bipartite $\cap$ planar [Thm. \ref{thm:NPH-planar}]};
        \node[draw,rectangle,rounded corners,red]  (bi) at (1.5cm, 2.5cm) {bipartite \cite{Hanaka:2019}};
        \node[draw,rectangle,rounded corners,blue]  (bp) at (1.5cm, 0cm) {bipartite permutation};
        \node[draw,rectangle,rounded corners,blue]  (pm) at (3.5cm, 0.9cm) {permutation};
        \node[draw,rectangle,rounded corners,red]  (comp) at (2cm, 4cm) {comparability};
        \node[draw,rectangle,rounded corners,blue]  (trap) at (4.7cm, 1.7cm) {$d$-trapezoid};
        \node[draw,rectangle,rounded corners,align=center,text width=2.4cm,blue]  (ccd) at (4.7cm, 3cm) {co-comparability of bounded dimension \cite{Bodlaender:1995}};
        \node[draw,rectangle,rounded corners,red]  (cc) at (4.7cm, 4.3cm) {co-comparability};
        \node[draw,rectangle,rounded corners,red]  (atf) at (4.7cm, 5.2cm) {AT-free};
        
        \node[draw,rectangle,rounded corners,blue]  (ch) at (8.2cm, 3.8cm) {chordal};
        \node[draw,rectangle,rounded corners,blue]  (int) at (7.2cm, 2.9cm) {interval};
        \node[draw,rectangle,rounded corners,align=center,blue]  (ca) at (7.2cm, 4.7cm) {circular-arc\\ \cite{Kloks:1998}};
        \node[draw,rectangle,rounded corners,align=center,text width=1.8cm,red] (cb) at (6.6cm, 0cm) {co-bipartite [Thm. \ref{thm:NPH-cobipartite}]};
        \node[draw,rectangle,rounded corners,align=center,blue]  (wch) at (9.2cm, 5cm) {weakly\\ chordal \cite{Bouchitte:2002}};
        \node[draw,rectangle,rounded corners,align=center,blue] (pig) at (8.5cm, 0cm) {proper\\ interval};
        \node[draw,rectangle,rounded corners,red]  (cf) at (8.5cm, 2cm) {claw-free};
        \node[draw,rectangle,rounded corners,align=center,red]  (line) at (10.3cm, 0cm) {line\\ {[Thm. \ref{thm:NPH-line}]}};
        \node[draw,rectangle,rounded corners,blue]  (cir) at (10.5cm, 4.2cm) {circle \cite{Kloks:1998}};
        \node[draw,rectangle,rounded corners,align=center,blue] (dh) at (10.3cm, 2.5cm) {distance\\ hereditary};
        
        \draw[->,>=stealth] (bi) -- (bpl);
        \draw[->,>=stealth] (pl) -- (bpl);
        \draw[->,>=stealth] (bi) -- (bp);
        \draw[->,>=stealth] (pm) -- (bp);
        \draw[->,>=stealth] (comp) -- (bi);
        \draw[->,>=stealth] (comp) -- (pm);
        \draw[->,>=stealth] (atf) -- (cc);
        \draw[->,>=stealth] (cc) -- (ccd);
        \draw[->,>=stealth] (ccd) -- (trap);
        \draw[->,>=stealth] (trap) -- (pm);
        \draw[->,>=stealth] (cc) [out=340] to (int);
        \draw[->,>=stealth] (ch) -- (int);
        \draw[->,>=stealth] (wch) -- (ch);
        \draw[->,>=stealth] (wch) -- (dh);
        \draw[->,>=stealth] (ca) -- (int);
        \draw[->,>=stealth] (cc) [out=340,in=97] to (cb); %bend
        \draw[->,>=stealth] (int) [out=270] to (pig);
        \draw[->,>=stealth] (cf) -- (pig);
        \draw[->,>=stealth] (cf) -- (line);
        \draw[->,>=stealth] (cf) -- (cb);
        \draw[->,>=stealth] (cir) -- (dh);
        \draw[->,>=stealth] (cir) [out=210,in=0] to (pm); %bend
        \draw[->,>=stealth] (cir) [out=210,in=47] to (pig); %bend
    \end{tikzpicture} 
    \caption{The figure depicts the hierarchy of classes of graphs. Directed arcs indicate the containment relation between two graph classes. Red rounded rectangles are classes for which {\sc Maximum Minimal Separator} is NP-hard and blue ones are those for which it is polynomial-time solvable.}
    \label{fig:complexity}
\end{figure}
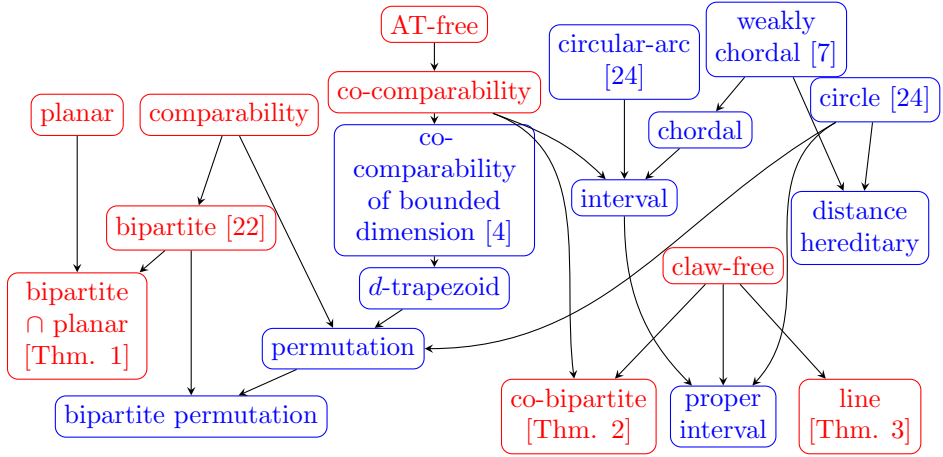

We first show that this problem is NP-complete even on several restricted graph classes.
\begin{theorem}\label{thm:NPH-planar}
     \textsc{Maximum Minimal Separator} is NP-complete even for subcubic planar bipartite graphs.
\end{theorem}
\begin{theorem}\label{thm:NPH-cobipartite}
    \textsc{Maximum Minimal Separator} is NP-complete even for co-bipartite graphs.
\end{theorem}
\begin{theorem}\label{thm:NPH-line}
    \textsc{Maximum Minimal Separator} is NP-complete even for line graphs.
\end{theorem}

Since we can enumerate all the minimal separators in polynomial time per output \cite{Berry:2000},
the problem is obviously tractable for the classes of graphs having a polynomial number of minimal separators.
These hardness and tractable results give an almost complete picture of the complexity landscape of {\sc Maximum Minimal Separator} on several graph classes (See Figure~\ref{fig:complexity}).

As for an algorithmic result, we design a $2^{O(k)}n^{O(1)}$-time algorithm for \textsc{Maximum Minimal Separator}, which even improves the previous claimed running time $2^{k^{O(1)}}n^{O(1)}$ by Hanaka et al.~\cite{Hanaka:2019}.
The algorithm we propose is inspired by algorithms for constructing tree decompositions due to Bodlaender et al.~\cite{Bodlaender:1997} and Skodinis \cite{Skodinis:1999}.

\begin{theorem}\label{thm:fpt}
    There is an algorithm for \textsc{Maximum Minimal Separator} that runs in time $2^{O(k)}n^{O(1)}$.
\end{theorem}

We also show that the asymptotic dependency of the exponential part in the running time cannot be drastically improved under the Exponential Time Hypothesis (ETH) \cite{Impagliazzo:2001}.

\begin{theorem}\label{thm:eth}
    Unless ETH fails, there is no $2^{o(k)}n^{O(1)}$-time algorithm for \textsc{Maximum Minimal Separator}.
\end{theorem}

It is well known that a parameterized problem is fixed-parameter tractable if and only if it admits a kernelization.
Given Theorem~\ref{thm:fpt}, we are interested in the size of a kernel for our problem, namely the polynomial kernelizability.
We, however, show that it is essentially impossible under some standard complexity-theoretic assumption.

\begin{theorem}\label{thm:kernel-general}
    {\sc Maximal Minimal Separator} does not admit a polynomial kernelization unless ${\rm NP} \subseteq {\rm coNP}/{\rm poly}$.
\end{theorem}

Since {\sc Maximal Minimal Separator} is trivially OR-compositional, that is, $G$ has a minimal separator of size at least $k$ if and only if some connected component of $G$ has, Theorem~\ref{thm:kernel-general}.

If input graphs are restricted to be connected, some OR-compositional problems can have polynomial kernelizations (see \cite{Fomin:2019}).
However, our problem is unlikely to have a polynomial kernelization.

\begin{theorem}\label{thm:kernel-connected}
    {\sc Maximum Minimal Separator for Connected Graphs} does not admit a polynomial kernelization unless ${\rm NP} \subseteq {\rm coNP}/{\rm poly}$.
\end{theorem}

\subsection{Paper organization}
The rest of this paper is organized as follows. 
Section \ref{sec:preli} is the preliminary section: we give some definitions and notations.
In Section \ref{sec:hardness}, we prove the NP-hardness of {\sc Maximal Minimal Separator} for several restricted graph classes.
In Section \ref{sec:fpt}, we give a fixed-parameter algorithm and lower bounds for {\sc Maximal Minimal Separator}. 
Section \ref{sec:conclusion} gives concluding remarks and future work.

\section{Preliminaries}\label{sec:preli}

Throughout the paper, graphs are simple and undirected.
Let $G = (V, E)$ be a graph. For $X \subseteq V$, we denote by $G[X]$ the subgraph induced by $X$.
For a vertex $v \in V$, the neighborhood of $v$ is denoted by $N(v)$, and for a vertex set $X \subseteq V$, we denote by $N(X) = \bigcup_{v \in X}N(v) \setminus X$. 

Let $S$ be a vertex set of $G$. A {\em full component associated to $S$} is a component $C$ of $G[V \setminus S]$ with $N(C) = S$.
The following folklore characterization of minimal separators is used throughout this paper.
\begin{lemma}[folklore]\label{lem:full-comp}
    A subset $S \subseteq V$ is a minimal separator if and only if there are at least two full components $C_1$ and $C_2$ associated to $S$.
    Moreover, if $S$ is a minimal separator, then $S$ is a minimal $a,b$-separator for every pair $a \in C_1$ and $b \in C_2$.
\end{lemma}

In particular, we can check if a subset $S \subseteq V$ is a minimal separator of $G$ in linear time.

For $S \subseteq V$, the pair $(S, V \setminus S)$ is called a {\em cut} of $G$.
The set of edges between $S$ and $V\setminus S$ is called the {\em cutset} of $(S, V\setminus S)$.
The size of a cut is defined as the number of edges of its cutset.
When both $G[S]$ and $G[V\setminus S]$ are connected, we say that $(S, V\setminus S)$ is connected.

\section{Hardness on graph classes}\label{sec:hardness}
\subsection{Planar bipartite graphs}\label{sec:planar}
In this section, we prove Theorem~\ref{thm:NPH-planar} by performing a polynomial-time reduction from the connected maximum cut problem on planar graphs, which is known to be NP-complete \cite{Haglin:1991}.
In this problem, given a graph $G = (V, E)$, the objective is to find a bipartition $(S, V\setminus S)$ of the vertex set $V$ maximizing the number of edges between $S$ and $V \setminus S$ subject to a connectivity requirement: both $G[S]$ and $G[V\setminus S]$ must be connected.
If we drop the connectivity requirement, the problem coincides with the well-known maximum cut problem.
Haglin and Venkatesan proved that this problem is NP-hard even for planar cubic graphs~\cite{Haglin:1991}, whereas the original maximum cut problem can be solved in polynomial time on planar graphs~\cite{Hadlock:1975}.

Let $G = (V, E)$ be a planar cubic graph. We subdivide each edge $e$ once by introducing a new vertex $v_e$, and we let $G'$ be the resulting graph. 
Clearly, $G'$ is bipartite and planar. 
\begin{lemma}\label{lem:planar-bipartite}
    $G$ has a connected cut of size at least $k$ if and only if $G'$ has a minimal separator of size at least $k$.    
\end{lemma}

\begin{proof}
Suppose first that $G$ has a connected cut $(S, V\setminus S)$ of size at least $k$.
Let $F$ be the cutset of $(S, V\setminus S)$ and
let $U = \{v_e : e \in F\}$.
Since $U$ separates $S$ and $V\setminus S$ in $G'$,  every vertex in $U$ is adjacent to vertices both in $S$ and in $V\setminus S$.
Therefore $U$ is a minimal separator of $G'$ of size at least $k$.

Conversely, let $U$ be a minimal separator of $G'$. We can assume that $U$ contains only vertices $v_e$'s that are newly introduced when subdividing each edge $e$.
To see this, consider a vertex $v \in U$ that is the original vertex in $G$. By the minimality of $U$, no neighbor of $v$ is in $U$.
Since $U$ is a minimal separator of $G'$, there are two full components $C_1$ and $C_2$ associated to $U$. As $v$ has degree three, we can assume that $v$ has exactly one neighbor $v_e$ in $C_1$ and at least one neighbor in $C_2$. Observe that there are two full components associated to $U \setminus \{v\} \cup \{v_e\}$: one is $C_1 \setminus \{v_e\}$ and the other one is the component $C$ containing $C_2$ and $v$. Note that $C$ is indeed a full component associated to $U \setminus \{v\} \cup \{v_e\}$ since $N_{G'}(C_2) = U$ and $v$ is adjacent to $v_e$. 
By repeatedly applying this to a minimal separator $U$, $U$ contains only newly introduced vertices $v_e$'s when subdividing edges.
Since each vertex of $U$ has degree two, there are only two full components $C_1$ and $C_2$ associated to $U$ and there are no components other than two. 
Now, we take an edge $e$ for each $v_e \in U$ and the edges taken here form the set of edges between $C_1 \cap V$ and $C_2 \cap V$ in $G$.
As $C_1 \cap V$ and $C_2 \cap V$ are connected, the lemma holds.
\end{proof}

The construction of $G'$ can be done in polynomial time and hence Theorem~\ref{thm:NPH-planar} follows.

\subsection{Co-bipartite graphs}\label{sec:cobipartite}
A graph $G$ is {\em co-bipartite} if the vertex set of $G$ can be partitioned into two subsets $A$ and $B$ so that both $G[A]$ and $G[B]$ induce cliques.
In other words, a graph is co-bipartite if and only if its complement is a bipartite graph.
In this section, we prove Theorem~\ref{thm:NPH-cobipartite}. 

Here, we prove that the problem of finding a maximum cardinality minimal separator is NP-hard even on co-bipartite graphs by giving a reduction from the minimum independent dominating set problem (equivalently, the minimum maximal independent set problem) on bipartite graphs, which is known to be NP-hard~\cite{Corneil:1984,Liu:2015}. 
In this problem, given a bipartite graph $G = (A \cup B, E)$ with bipartition $(A, B)$ of the vertex set, the goal is to find a minimum cardinality maximal independent set of $G$.
Our reduction consists of several types of ``complement'' operations on graphs and solutions.

From an instance $G = (A \cup B, E)$ of the minimum independent domination set problem, we take the bipartite complement $G' = (A \cup B, (A \times B) \setminus E)$. It is easy to see that $U \subseteq A \cup B$ is a maximal independent set in $G$ if and only if $U$ is a maximal biclique in $G'$. Therefore, we have the following intermediate result.
\begin{lemma}\label{lem:mmbiclique}
    The problem of finding a minimum cardinality maximal biclique is NP-hard on bipartite graphs.
\end{lemma}

Let $G'' = (A \cup B, \overline{E})$ be the co-bipartite graph that is the complement graph of $G'$: $G''$ contains two cliques induced by $A$ and $B$ and there is an edge between $a \in A$ and $b \in B$ in $G''$ if and only if $a$ is {\em not} adjacent to $b$ in $G'$.
In the following, we assume that $G'$ has no $v \in A$ with $N(v) = B$ and no $v \in B$ with $N(v) = A$ since such a vertex always belongs to any maximal biclique of $G'$.
We also assume that $G'$ has no isolated vertices.

\begin{lemma}\label{lem:co-bipartite}
    Let $U \subseteq A \cup B$.
    Then $U$ is a maximal biclique of $G'$ that contains at least one vertex from $A$ and one from $B$ if and only if $(A \cup B) \setminus U$ is a minimal separator in $G''$.
\end{lemma}

\begin{proof}
Suppose $U$ is a maximal biclique in $G'$ with $A \cap U \neq \emptyset$ and $B \cap U \neq \emptyset$. Then, there are no edges between $A \cap U$ and $B \cap U$ in $G''$. Thus, $(A \cup B) \setminus U$ is a separator that separates $A \cap U$ and $B \cap U$ in $G''$. Moreover, there are exactly two components in $G''[(A \cap U) \cup (B \cap U)]$.
We claim that such components are full components associated to $(A \cup B) \setminus U$. 
To see this, consider a vertex $v \in A \setminus U$. Since $A$ is a clique in $G''$, every vertex in $A$ is adjacent to $v$. 
Moreover, there is $b \in B \cap U$ that is adjacent to $v$ since otherwise we can add $v$ into $U$ and obtain a biclique $U \cup \{b\}$ of $G'$, contradicting to the maximality of $U$. 
Therefore, $v$ has a neighbor in both $A \cap U$ and $B \cap U$, which implies $A \cap U$ and $B \cap U$ are full components associated to $(A \cup B) \setminus U$.

Conversely, let $S = (A \cup B) \setminus U$ be a minimal separator.
We first observe that $(A \cup B) \setminus S$ contains at least one vertex of $A$ and at least one from $B$.
As $S$ is a separator of $G''$, there are no edges between $A \setminus S$ and  $B \setminus S$, which implies $(A \cup B) \setminus S$ is a biclique in $G'$.
Moreover, every vertex $v \in S$ has a neighbor both in $A \setminus S$ and in $B \setminus S$ in $G''$. This means that $((A \cup B) \setminus S) \cup \{v\}$ does not induce a biclique of $G'$.
Therefore, $(A \cup B) \setminus S = U$ is a maximal biclique in $G'$.
\end{proof}

Therefore, finding a minimum cardinality maximal biclique in $G'$ is equivalent to finding a minimal separator in $G''$. Hence, Theorem~\ref{thm:NPH-cobipartite} follows.

\subsection{Line graphs}\label{sec:line}

A line graph $G'$ of a graph $G = (V, E)$ is a graph with vertex set $E$ such that every pair of vertices $e$ and $f$ with $e, f \in E$ are adjacent to each other in $G'$ if and only if $e$ is incident to $f$ in $G$.

Let $G = (V, E)$ be a connected graph and $G' = (V_E = \{v_e : e \in E\}, F)$ be its line graph.
We say that a cut $(S, V \setminus S)$ is {\em non-trivial} if both $S$ and $V \setminus S$ contain at least two vertices. 

\begin{lemma}\label{lem:line}
    $G$ has a non-trivial connected cut $(S, V \setminus S)$ of size at least $k$ if and only if $G'$ has a minimal separator of size at least $k$.
\end{lemma}

\begin{proof}
    Suppose that $V$ can be partitioned into $(S, V\setminus S)$ such that both $G[S]$ and $G[V \setminus S]$ are connected and have at least two vertices.
    Let $F$ be the set of edges between $S$ and $V \setminus S$.
    We claim that the set $V_F = \{v_e : e \in F\}$ of vertices corresponding to $F$ forms a minimal separator in $G'$.
    To see this, consider a vertex $v_e \in V_F$. Let $E_1$ be the set of edges in $G[S]$ and $E_2$ the set of edges in $G[V \setminus S]$. 
    By the connectivity of $G[S]$ and $G[V \setminus S]$, the vertex sets $C_1 = \{v_e \in V_E: e \in E_1\}$ and $C_2 = \{v_e \in V_E: e \in E_2\}$ are connected in $G'$.
    Moreover, there are no edges between $C_1$ and $C_2$.
    Thus, $V_F$ separates $C_1$ and $C_2$ in $G'$.
    Now, let $v_e \in V_F$. As both $E_1$ and $E_2$ are not empty, at least one edge both in $E_1$ and in $E_2$ is incident to $e$ in $G$.
    This implies that $v_e$ has a neighbor both in $C_1$ and in $C_2$.
    Therefore, $C_1$ and $C_2$ are full components associated to $V_F$.
    
    Conversely, let $S$ be a minimal separator of $G'$ of size at least $k$. 
    We first show that $G'[V \setminus S]$ has exactly two components.
    Since every line graph is a claw-free graph, $G'$ is indeed claw-free.
    Suppose for contradiction that $G'[V \setminus S]$ has at least three components, say $C_1, C_2, C_3$.
    As $S$ is a minimal separator of $G'$, we can assume that $C_1$ and $C_2$ are full components associated to $S$.
    Let $v_3$ be a vertex in $C_3$ that has a neighbor $w$ in $S$.
    Since $S$ is a minimal separator of $G'$, $w$ has a neighbor $v_1$ in $C_1$ and $v_2$ in $C_2$.
    The four vertices $v_1, v_2, v_3, w$ induce a claw in $G'$, contradicting to the fact that $G'$ is claw-free.
    Thus, $S$ separates exactly two components $C_1$ and $C_2$ in $G'$, and both are full components associated to $S$.
    
    Let $C_1$ and $C_2$ be the full components associated to $S$ of $G'$ and let $V_1$ (resp. $V_2$) be the set of end vertices of edges $E_1 = \{e \in E: v_e \in C_1\}$ (resp. $E_2 = \{e \in E: v_e \in C_2\}$) of $G$.
    Observe that $V_1$ and $V_2$ are disjoint as otherwise two edges $e_1$ in $E_1$ and $e_2$ in $E_2$ are incident to $v \in V_1 \cap V_2$, which implies $v_{e_1}$ and $v_{e_2}$ are adjacent in $G'$.
    Moreover, if some vertex $v \in V$ is not in $V_1 \cup V_2$, then every vertex corresponding to an edge incident to $v$ belongs to $S$, which contradicts to the fact that $C_1$ and $C_2$ are only full components associated to $S$.
    Thus, $(V_1, V_2)$ is a bipartition of $V$.
    
    To see the connectivity of $V_1$ (and $V_2$), let us consider an arbitrary pair of vertices $u, v \in V_1$.
    By the definition of $V_1$, every vertex of $V_1$ has at least one edge in $E_1$ that is incident to it.
    Let $e_u$ and $e_v$ be edges in $E_1$ incident to $u$ and $v$, respectively.
    As $C_1$ is a connected component of $G'$, there is a path between the vertices correspond to $e_u$ and $e_v$ in $G'$.
    The vertices on this path also form the edges of a path in $G$ and hence there is a path between $u$ and $v$ in $G$.
    Therefore, as $V_1$ and $V_2$ are connected, $(V_1, V_2)$ is a connected cut of $G$.
    
    Finally, if an edge $e$ lies between two set $V_1$ and $V_2$, then $v_e$ belongs to $S$, which implies that the size of the cut $(V_1, V_2)$ is exactly $|S|$.
\end{proof}

As we have seen in Section~\ref{sec:planar}, the connected maximum cut problem is NP-hard \cite{Haglin:1991,Duarte:2019,Eto2019}. 
We perform a polynomial-time reduction from this problem to our problem.
Given a graph $G = (V, E)$, we add a pendant vertex for each vertex of $G$.
Then, we claim that $G$ has a connected cut of size at least $k > 1$ if and only if the obtained graph $G' = (V', E')$ has a non-trivial connected cut of size at least $k > 1$.
To see this, consider a connected cut $(S, V' \setminus S)$ of $G'$.
Observe that for each $v \in V$, $v$ and its pendant neighbor are not separated in the cut.
This follows from the connectivity of $G[S]$ and $G[V' \setminus S]$.
Therefore, the problem remains hard even if we restrict ourselves to finding an optimal non-trivial cut.
By Lemma~\ref{sec:line}, the graph has a non-trivial connected cut
of size at least $k$ if and only if its line graph has a minimal separator of size at least $k$.

% \subsection{Inapproximability}
% Duarte et al.~\cite{DLPSS19} recently proved that there is no constant factor approximation algorithms for the connected maximum cut problem unless P$=$NP.
% Since the reduction in the proof of Lemma \ref{lem:line} preserves the optimal value, the inapproximability is inherited from the connected maximum cut problem.
% \begin{theorem}\label{thm:inapprox:line}
%     Unless P$=$NP, there is no constant factor approximation algorithm for {\sc Maximum Minimal Separator} even on line graphs. 
% \end{theorem}

% Furthermore, we show that {\sc Maximum Minimal Separator} admits no constant factor approximation algorithm even on bipartite graphs.
% The proof is similar to Theorem \ref{thm:NPH-planar}.
% Given an instance $(G=(V,E),k)$ of the connected maximum cut problem, we subdivide each edge $e$ by a vertex $v_e$.
% Then we  add a pendant vertex for each original vertex in $V$.
% Let $G'$ be the constructed graph.
% Without loss of generality, we assume that $k>1$.
% It is easily seen that all pendant vertices and original vertices are not contained in any minimal separator of size at least $k$ in $G'$. 
% Thus, we have the fact that there is a connected maximum cut of size at least $k$ in $G$ if and only if there is a minimal separator of size at least $k$ in $G'$. 
% \begin{theorem}\label{thm:inapprox:bipartite}
%     Unless P$=$NP, there is no constant factor approximation algorithm for {\sc Maximum Minimal Separator} even on bipartite graphs. 
% \end{theorem}

\section{Fixed-parameter algorithm and lower bounds}\label{sec:fpt}

\subsection{FPT algorithm}
This subsection is devoted to giving a $2^{O(k)}n^{O(1)}$-time algorithm for \textsc{Maximum Minimal Separator}.
Moreover, if the answer is affirmative, the algorithm computes such a minimal separator.
Before describing our algorithm, we need several definitions and known facts.

A {\em tree decomposition} of $G = (V, E)$ is a tree $T$ with node set $I$ and each node $v$ of $T$ is associated to a subset $X_v$ of $V$, called a {\em bag}, such that (1) $\bigcup_{v \in I}X_v = V$, (2) for each $e \in E$, there is $v \in I$ such that $e \subseteq X_v$, and (3) for each $x \in V$, the bags containing $x$ form a subtree of $T$.
The {\em width} of the tree decomposition $T$ is the maximum size of a bag minus one, and the {\em treewidth} of $G$ is the minimum integer $w$ such that $G$ has a tree decomposition of width $w$.
It is known that tree decompositions and minimal separators are deeply related to each other.
In particular, our algorithm is based on the well-known recursive construction of a tree decomposition using minimal separators.

Let $G = (V, E)$ be a graph and $S$ be a minimal separator of $G$.
Let $C$ be a (not necessarily full) component associated to $S$.
We denote by $G(C, S)$ the graph obtained from $G[C \cup S]$  by completing $S$ into a clique.

\begin{lemma}[Lemma 5 in \cite{Skodinis:1999}]\label{lem:hereditary}
    Let $S'$ be a minimal separator of $G(C, S)$. Then, it is also a minimal separator of $G$.
\end{lemma}

A minimal separator of $G$ is called a {\em clique minimal separator} if it induces a clique in $G$.
\begin{lemma}[Property 3.2 in \cite{Berry:2010}]\label{lem:cms}
    Let $S$ be a clique minimal separator of $G$ and $C$ be a full component associated to $S$.
    Then, every minimal separator other than $S$ is a minimal separator of $G[C \cup S]$ or of $G[V \setminus C]$. 
\end{lemma}

Finally, Hanaka et al. \cite{Hanaka:2019} proved that finding a maximum cardinality minimal separator is tractable for bounded treewidth graphs.

\begin{theorem}[Corollary 4.14 in \cite{Hanaka:2019}]\label{thm:tw-dp}
    Given a tree decomposition of $G$ of width $w$,
    one can find a maximum cardinality minimal separator of $G$ in time $2^{O(w)}n^{O(1)}$ if it exists.
\end{theorem}

The idea of our algorithm appeared partially in~\cite{Bodlaender:1997} and \cite{Skodinis:1999}.
We can find a minimal separator $S$ in $G$ in polynomial time by taking an arbitrary separator and greedily removing vertices until it becomes minimal.
If the cardinality of $S$ is at least $k$, we are clearly done in this case.
Suppose otherwise.
Let $C_1, C_2, \ldots, C_t$ be the components of $G[V \setminus S]$.
We recursively apply this process to $G(C_i, S)$ for each $C_i$.
Hopefully, this recursive algorithm may compute a minimal separator of size at least $k$ or a tree decomposition of width at most $k$.
In these cases, we can solve {\sc Maximum Minimal Separator} by using Lemma~\ref{lem:hereditary} or Theorem~\ref{thm:tw-dp}.
However, the main obstacle here is that the algorithm may fail to find a minimal separator since $G(C_i, S)$ may form a clique even if $G[C_i \cup S]$ is not a clique since in the process of our algorithm, we add some edges not appeared in the original graph $G$.
If such cliques are all small, we are able to construct a tree decomposition of small width as well.
Otherwise, we can still conclude that this large clique contains either a large minimal separator or a clique minimal separator of the original graph $G$.
In the latter case, by Lemma~\ref{lem:cms}, we can safely decompose $G$ into two or more subgraphs by this clique minimal separator. 

Now, we formally describe our recursive algorithm. Let $G = (V, E)$ be the input graph.
Without loss of generality, we assume that $G$ is not a complete graph as otherwise there is no minimal separator in $G$.
The main procedure {\sc FindSep($G$, $S$, $k$)} is as follows.

\bigskip\noindent
{\sc FindSep($H = (U, F)$, $S$, $k$)}:
\begin{itemize}
    \item[] \hspace{-0.35cm}{\bf Invariants}: $S \subseteq U$ and $|S| < k$. 
    \item[1] If $H$ has no more than $2k - 1$ vertices, we do nothing.
    \item[2] Otherwise, find an arbitrary minimal separator $S'$ of $H$.
    \begin{itemize}
        \item[2-1] If $S'$ is found and $|S'| \ge k$, report ``YES'' and halt.
        \item[2-2] If $S'$ is found and $|S'| < k$, call {\sc FindSep($H(C_i, S')$, $S'$, $k$)} for each component $C_i$ of $H[U \setminus S']$.
        \item[2-3] Suppose there is no (minimal) separator in $H$ (i.e. $H$ is a complete graph).
        \begin{itemize}
            \item[2-3-1] If $U$ does not induce a clique in the original graph $G$, report ``YES'' and halt.
            \item[2-3-2] Suppose $U$ induces a clique in the original graph $G$. Let $C$ be a connected component of $G[V \setminus S]$. Then, output $G[C \cup S]$ and $G[V \setminus C]$.
        \end{itemize}
    \end{itemize}
\end{itemize}

If we call {\sc FindSep($G$, $\emptyset$, $k$)}, there are several outcomes.
Suppose first that case (2-3) never happens during the execution of {\sc FindSep}($G$, $\emptyset$, $k$) and subsequent recursive calls.
If the algorithm reports ``YES'' in (2-1), clearly there is a minimal separator of size at least $k$ in $H$.
By Lemma~\ref{lem:hereditary}, this separator is also a minimal separator of $G$, and hence we are done.
Otherwise, we can conclude that $G$ has a tree decomposition of width at most $2k-2$.
To see this, we use the following well-known fact.

\begin{lemma}[e.g. \cite{Bodlaender:1993}]\label{lem:clique-bag}
    Let $T$ be a tree decomposition of $G$. Then, for every clique $K$ in $G$, there is a bag in $T$ containing $K$.
\end{lemma}

We can inductively construct a tree decomposition of $H$ as follows.
If $H$ has at most $2k - 1$ vertices, we output a single bag that has all vertices of $H$.
Otherwise, we can find a minimum separator $S'$ of size at most $k - 1$. By applying induction to each $H(C_i, S')$, we have a tree decomposition $T_i$ of $H(C_i, S')$ of width at most $2k - 2$. 
By Lemma~\ref{lem:clique-bag}, $T_i$ has a bag $B_i$ that entirely contains $S'$.
We introduce a new bag $S'$ and construct a tree decomposition by connecting $B_i$ to $S'$ for each $T_i$. 
As $|S'| < k$, the obtained decomposition has width at most $2k - 2$ as well. 
Owing to Theorem~\ref{thm:tw-dp}, we can find a maximum cardinality minimal separator of $G$ in this case.

Now, we will see the validity of the case (2-3). Suppose that $H$ is a complete graph obtained in the execution of  {\sc FindSep}.
Let $K$ be the set of vertices of $H$. Recall that $|K| \ge 2k$.
Note that this clique may not induce a clique in the original graph since we add some edges during the execution of {\sc FindSep}.
However, if $K$ is not a clique in $G$, we can always find a large minimal separator inside $K$.

\begin{lemma}\label{lem:large-clique}
    Let $K$ be defined as above.
    If $G[K]$ is not a clique, $K$ contains a minimal separator of $G$ of size at least $k$.
\end{lemma}
\begin{proof}
Let $G_0 = G$, $G_j = G_{j-1}(C_{j-1}, S_{j-1})$ for $1 \le j < m$, where $S_j$ is a minimal separator of $G_j$ and $C_j$ is a component of $G_j[V \setminus S_j]$, and $G_m = H$ be the graphs appeared in the path of the search tree of {\sc FindSep} between the root $G$ and the leaf $H$.
Let $j < m$ be the maximum index such that $G_{j}[K]$ is not a clique.
Let $u$ and $v$ be two vertices not adjacent to each other in $G_{j}[K]$.
Since $G_{j+1}$ is of the form $G_j(C_j, S_j)$ and $K$ is a clique in $G_{j+1}$, every vertex of $K \setminus S_j$ is adjacent to both $u$ and $v$.
Recall that $S_j$ has less than $k$ vertices.
This implies that $u$ and $v$ have at least $k$ common neighbors in $G_{j}$.
Since every minimal $u, v$-separator must contain all the common neighbors, $G_{j}$ contains a minimal separator of size at least $k$.
By Lemma~\ref{lem:hereditary}, it holds that $G_i$ has a minimal separator of size at least $k$ for all $0 \le i \le j$.
\end{proof}

The proof of Lemma~\ref{lem:large-clique} allows us to efficiently find a minimal separator of size at least $k$, which is contained in $K$, for this case.

Suppose otherwise that $K$ induces a clique in $G$.
Since $G$ is not a complete graph, $H$ is of the form $G_{m - 1}(C_{m-1}, S_{m-1})$ as in the proof of Lemma~\ref{lem:large-clique}. Since $S_{m-1}$ is a minimal separator of $G_{m - 1}$ and hence so is in $G$.
This means that $K$ contains at least one minimal separator $S := S_{m-1}$ of $G$.
A crucial observation is that $S$ is a clique minimal separator of $G$ with size at most $k - 1$.
Therefore, by Lemma~\ref{lem:cms}, every minimal separator of size at least $k$ of $G$ appears in either $G[C \cup S]$ or $G[V \setminus C]$ if it exists, where $C$ is a component of $G[V \setminus S]$.
We summarize the above discussion in the following lemma.

\begin{lemma}\label{lem:findsep}
    Let $G$ be a non-complete graph.
    If we call {\sc FindSep($G$, $\emptyset$, $k$)}, it produces either
    \begin{itemize}
        \item a minimal separator of size at least $k$,
        \item a tree decomposition of width at most $2k - 2$, or
        \item two induced subgraphs $G'$ and $G''$ of $G$ with $|V(G')| < |V(G)|$ and $|V(G'')| < |V(G)|$ such that $G$ has a minimal separator of size at least $k$ if and only if at least one of $G'$ and $G''$ has.
    \end{itemize}
\end{lemma}

Let $G$ be a graph given as input and let $k$ be a positive integer. If $G$ is a complete graph, we can immediately conclude that $G$ has no minimal separator of size at least $k$.
Otherwise, we call {\sc FindSep}($G$, $\emptyset$, $k$) and, by Lemma~\ref{lem:findsep}, obtain either a minimal separator of size at least $k$, a tree decomposition of width at most $2k - 2$, or two induced subgraphs $G'$ and $G''$ with $|V(G')| < |V(G)|$ and $V(G'') < |V(G)|$ such that $G$ has a minimal separator of size at least $k$ if and only if at least one of $G'$ and $G''$ has. If the first outcome occurs, we are done. If the second outcome occurs, we apply the dynamic programming algorithm in Theorem~\ref{thm:tw-dp} on the obtained tree decomposition of width at most $2k - 2$. 
Finally, if the third outcome occurs, we recursively apply the whole algorithm to $G'$ and $G''$, that is, call {\sc FindSep}($G'$, $\emptyset$, $k$) and {\sc FindSep}($G''$, $\emptyset$, $k$).
Since $G$ has a minimal separator of size at least $k$ if and only if at least one of $G'$ or $G''$ has, this recursive application correctly finds a minimal separator of size at least $k$.

Finally, we estimate the running time of the entire algorithm.
Observe that {\sc FindSep} runs in polynomial time and the algorithm in Theorem~\ref{thm:tw-dp} runs in $2^{O(k)}n^{O(1)}$ time,
we can find a minimal separator of size at least $k$ within the claimed running time if it exists for the first and second outcomes.
Since the graphs $G'$ and $G''$ in the third outcome are not necessarily disjoint, the running time could be exponential in $n$ at a first glance.
However, the following property of clique minimal separators ensures that the overall running time is still $2^{O(k)}n^{O(1)}$.

\begin{lemma}[\cite{Berry:2010}]
    Every clique minimal separator of $G$ is also a minimal separator in any minimal triangulation of $G$.
\end{lemma}
A {\em triangulation} of $G = (V, E)$ is a chordal super graph $H = (V, E')$ satisfying $E \subseteq E'$.
A triangulation is {\em minimal} if there is no triangulation $H' = (V, E'')$ of $G$ such that $E \subseteq E'' \subset E'$.
Since every $n$-vertex chordal graph has $O(n)$ minimal separators, $G$ can have only $O(n)$ clique minimal separators.
Moreover, by Lemma~\ref{lem:hereditary}, every clique minimal separator found as the third outcome is also a clique minimal separator of $G$.
This implies the third outcome only occurs $O(n)$ times in the entire execution.
Therefore, the overall running time is still bounded by $2^{O(k)}n^{O(1)}$.

It would be worth mentioning that our algorithm also works in the vertex-weighted setting discussed in \cite{Hanaka:2019}.
\begin{corollary}
    Given a vertex-weighted graph $G$ with $w: V \to \mathbb N_{>0}$ and an integer $k$, one can determine $G$ has a minimal separator of weight at least $k$ in time $2^{O(k)}n^{O(1)}$. Moreover, if the answer is affirmative, the algorithm outputs such a minimal separator in the same running time. 
\end{corollary}

\subsection{Lower bound based on ETH}
Impagliazzo et al. \cite{Impagliazzo:2001} proved that there is no $2^{o(n+m)}$-time algorithm for {\sc $3$-CNFSAT} assuming that the Exponential Time Hypothesis~\cite{Impagliazzo:2001}, where $n$ is the number of variables and $m$ is the number of clauses in the input formula. 
From this starting point, a lot of complexity lower bounds have been established in the literature. (See \cite{Lokshtanov:2011}, for example.)

In this subsection, we verify that a known chain of reductions from {\sc $3$-CNFSAT} to {\sc Maximum Minimal Separator} proves Theorem~\ref{thm:eth}.
To show Theorem~\ref{thm:eth}, under ETH, it suffices to show that there is no $2^{o(n)}$-time algorithm for {\sc Maximum Minimal Separator} as $k \le n$.

We begin with the following well-known results.

\begin{lemma}[e.g. Theorem 14.6 in \cite{Cygan:2015}]\label{lem:vc}
    Let $\phi$ be a 3-CNF formula with $n$ variables and $m$ clauses.
    Then, there is a polynomial-time algorithm that constructs a graph $G$ with $O(n+m)$ vertices and $O(n+m)$ edges such that
    $\phi$ is satisfiable if and only if $G$ has a dominating set of size at most $k$ for some $k$.
\end{lemma}

As we have mentioned in Section~\ref{sec:cobipartite}, the problem of finding a minimum cardinality independent dominating set problem is NP-hard even on bipartite graphs.
More specifically, Corneil and Perl~\cite{Corneil:1984} proved the following lemma.

\begin{lemma}[\cite{Corneil:1984}]
    Let $G$ be a graph with $n$ vertices and $m$ edges.
    Let $G'$ be a bipartite graph obtained from $G$ by replacing each edge with a path of five vertices.
    Then, $G$ has a dominating set of size at most $k$ if and only if $G'$ has an independent dominating set of size at most $m + k$.
\end{lemma}

Plugging the above chain of reductions into the polynomial-time reduction described in Section~\ref{sec:cobipartite},
we can construct in polynomial time a graph $G$ with $O(n+m)$ vertices from an instance of {\sc $3$-CNFSAT} with $n$ variables and $m$ clauses such that $\phi$ is satisfiable if and only if $G$ has a minimal separator of size at least $k$ for some $k$.
This means that there is no $2^{o(n)}n^{O(1)}$-time algorithm for {\sc Maximum Minimal Separator}, where $n$ is the number of vertices of the input graph, unless ETH fails.
Therefore, Theorem~\ref{thm:eth} follows.

\subsection{Kernel lower bound}\label{ssec:kernel}

For a parameterized (decision) problem $P$ with instance $I$ and parameter $k$, a {\em kernelization} is a polynomial-time preprocessing that outputs an equivalent pair $(I', k')$, such that $(I, k)$ is a YES-instance if and only if so is $(I', k')$ and $|I'| + k' \le f(k)$ for some computable function $f$. 
In particular, if $f$ is polynomial, it is called a {\em polynomial kernelization}.

If $G$ has more than one connected component, then every minimal separator is contained in its components as a minimal separator.
Hence, there is a trivial OR-composition \cite{Bodlaender:2009} from {\sc Maximum Minimal Separator} into itself on not necessarily connected graphs.
With the result of Bodlanender et al. \cite{Bodlaender:2009}, a polynomial kernelization is unlikely to exist, and hence Theorem~\ref{thm:kernel-general} follows.

This argument essentially requires that input graphs are disconnected.
One may expect that if the input graph is restricted to be connected, there could be a polynomial kernelization.
However, such an expectation is unlikely.
To see this, we show the following lemma.

\begin{lemma}\label{lem:con-composition}
    Let $G_1, G_2, \ldots, G_t$ be graphs with $G_i = (V_i, E_i)$ for each $1 \le i \le t$ and let $H = (V, E)$ be the graph obtained from these $t$ graphs by adding a universal vertex $r$, that is, $r$ is adjacent to every vertex in $V_1 \cup \cdots \cup V_t$.
    Then, at least one of these $t$ graphs has a minimal separator of size at least $k$ if and only if $H$ has a minimal separator of size at least $k + 1$.
\end{lemma}

\begin{proof}
    Suppose that at least one of $t$ graphs, say $G_1$, contains a minimal separator $S$ of size at least $k$.
    Then, there are two full components $C_1$ and $C_2$ in $G[V_1 \setminus S]$.
    Since $r$ has a neighbor both in $C_1$ and $C_2$ and there are no edges between them, $C_1$ and $C_2$ are full components of $H[V \setminus (S \cup \{r\})]$.
    Thus $H$ has a minimal separator of size at least $k + 1$.
    
    Conversely, $H$ has a minimal separator $S$ of size at least $k + 1$.
    Since $r$ is a universal vertex, it must be contained in $S$.
    Let $C_1$ and $C_2$ be full components of $H[V \setminus S]$.
    Since $r$ is a separator of $H$, $C_1$ is contained in $V_i$ for some $1 \le i \le t$.
    Moreover, $S \setminus \{r\}$ is also contained in $V_i$ as every vertex in $S \setminus \{r\}$ has a neighbor in $C_1$.
    We apply this argument to $C_2$ and hence $C_1$, $C_2$, and $S$ are all contained in $V_i$.
    Therefore, $S \setminus \{r\}$ is a minimal separator of $G_i$ with full components $C_1$ and $C_2$.
\end{proof}

Therefore, we have Theorem~\ref{thm:kernel-connected}.

\section{Concluding remarks}\label{sec:conclusion}
In this paper, we investigate the computational complexity of {\sc Maximum Minimal Separator} with respect to graph classes.
More concretely, we show that the problem is NP-complete even if the input is restricted to cubic planar bipartite, co-bipartite, and line graphs.
We also give an FPT algorithm for finding a minimal separator of size at least parameter $k$, whose exponential dependency is asymptotically optimal under the Exponential Time Hypothesis (ETH).

There are several interesting questions related to our results.
It is worth noting that the graph classes indicated by blue color in Figure~\ref{fig:complexity} has polynomially many minimal separators.
We have known that graph classes that have exponentially many minimal separators but for which {\sc Maximum Minimal Separator} can be solved in polynomial time are of bounded-treewidth proved by \cite{Hanaka:2019}.
More generally, the property of being a minimal separator can be expressed by a formula in MSO$_1$:
The property of being an $a$-$b$ minimal separator can be expressed as:
\begin{align*}
\phi_{a,b}(S) &:= \exists A, B \subseteq V.(a\in A \land b\in B\land A \cap B = \emptyset \land A \cap S = \emptyset \land B \cap S = \emptyset  \\
&\land {\bf fullcomp}(S, A) \land {\bf fullcomp}(S, B)
\land [\forall a' \in A, \forall b' \in B.(\neg{\bf adj}(a', b'))],\\
{\bf fullcomp}(S, C) &:= {\bf comp}(S, C) \land \forall v\in S.(\exists w\in C.({\bf adj}(v,w))),\\
{\bf comp}(S, C) &:= {\bf conn}(C) \land \forall v \in V.(v \notin S \cup C \implies \forall w \in C.(\neg{\bf adj}(v, w))),
\end{align*}
where {\bf conn}$(X)$ is the predicate that is true if and only if $G[X]$ is connected, and then the property of being a minimal separator can be expressed as:
\begin{eqnarray*}
    \phi(S) := \exists a, b \in V.(a\neq b \land \phi_{a,b}(S)).
\end{eqnarray*}
Therefore, {\sc Maximum Minimal Separator} is fixed-parameter tractable parameterized by cliquewidth via Courcelle's theorem for bounded-cliquewidth graphs~\cite{Courcelle:2000}.
%Since there are single-exponential algorithm for {\sc Maximum Minimal Separator} parameterized by treewidth and modular-width, respectively, it is worth to design a non-trivial single-exponential algorithm parameterized by cliquewidth.
It would be interesting to seek non-trivial graph classes having exponentially many minimal separators but for which {\sc Maximum Minimal Separator} can be solved in polynomial time, which could give a new insight for problems related to minimal separators, such as {\sc Treewidth} and {\sc Minimum Fill-in}.

Another stimulating open problem would be the applicability of Skodinis's FPT algorithms parameterized by the size of a maximum minimal separator.
He showed that several NP-hard problems, such as {\sc Maximum Independent Set} and {\sc Graph Coloring} can be solved in time $f(k)n^{O(1)}$ if every minimal separator of the input graph has size at most $k$ \cite{Skodinis:1999}.
He also showed that {\sc Hamiltonian Circuit} is NP-complete even on graphs having minimal separators of size at most three.
It would be interesting to draw a complexity-theoretic boundary of problems that are tractable on bounded-treewidth graphs but are intractable on graphs having only bounded-size minimal separators.

 \bibliographystyle{plain}

\bibliography{ref}

\begin{thebibliography}{10}

\bibitem{Berry:2000}
Anne Berry, Jean-Paul Bordat, and Olivier Cogis.
\newblock Generating all the minimal separators of a graph.
\newblock {\em Int. J. Foundations of Comput. Sci.}, 11(03):397--403, 2000.

\bibitem{Berry:2010}
Anne Berry, Romain Pogorelcnik, and Geneviève Simonet.
\newblock An introduction to clique minimal separator decomposition.
\newblock {\em Algorithms}, 3(2):197--215, 2010.

\bibitem{Bodlaender:2009}
Hans~L. Bodlaender, Rodney~G. Downey, Michael~R. Fellows, and Danny Hermelin.
\newblock On problems without polynomial kernels.
\newblock {\em Journal of Computer and System Sciences}, 75(8):423-- 434, 2009.

\bibitem{Bodlaender:1995}
Hans~L. Bodlaender, Ton Kloks, and Dieter Kratsch.
\newblock Treewidth and pathwidth of permutation graphs.
\newblock {\em SIAM J. Discret. Math.}, 8(4):606---616, 1995.

\bibitem{Bodlaender:1993}
Hans~L. Bodlaender and Rolf~H. Möhring.
\newblock The pathwidth and treewidth of cographs.
\newblock {\em SIAM Journal on Discrete Mathematics}, 6(2):181--188, 1993.

\bibitem{Bodlaender:1997}
Hans~L. Bodlaender, Jan van Leeuwen, Richard Tan, and Dimitrios~M. Thilikos.
\newblock On interval routing schemes and treewidth.
\newblock {\em Information and Computation}, 139(1):92--109, 1997.

\bibitem{Bouchitte:2002}
Vincent Bouchitt\'{e} and Ioan Todinca.
\newblock Treewidth and minimum fill-in: Grouping the minimal separators.
\newblock {\em SIAM Journal on Computing}, 31(1):212--232, 2001.

\bibitem{Chuzhoy:2019}
Julia Chuzhoy and Zihan Tan.
\newblock Towards tight(er) bounds for the excluded grid theorem.
\newblock In {\em Proceedings of the Thirtieth Annual ACM-SIAM Symposium on
  Discrete Algorithms}, pages 1445--1464, 2019.

\bibitem{Corneil:1984}
D.G. Corneil and Y.~Perl.
\newblock Clustering and domination in perfect graphs.
\newblock {\em Discrete Applied Mathematics}, 9(1):27-- 39, 1984.

\bibitem{Courcelle:2000}
B.~Courcelle, J.~A. Makowsky, and U.~Rotics.
\newblock Linear time solvable optimization problems on graphs of bounded
  clique-width.
\newblock {\em Theory of Computing Systems}, 33(2):125--150, 2000.

\bibitem{Cygan:2015}
Marek Cygan, Fedor~V. Fomin, Lukasz Kowalik, Daniel Lokshtanov, Daniel Marx,
  Marcin Pilipczuk, Michal Pilipczuk, and Saket Saurabh.
\newblock {\em Parameterized Algorithms}.
\newblock Springer Publishing Company, Incorporated, 1st edition, 2015.

\bibitem{Dirac:1961}
G.~A. Dirac.
\newblock On rigid circuit graphs.
\newblock {\em Abhandlungen aus dem Mathematischen Seminar der Universit{\"a}t
  Hamburg}, 25:71--76, 1961.

\bibitem{Duarte:2019}
Gabriel~L. Duarte, Daniel Lokshtanov, Lehilton L.~C. Pedrosa, Rafael C.~S.
  Schouery, and U{\'{e}}verton~S. Souza.
\newblock Computing the largest bond of a graph.
\newblock In {\em Proceedings of the 14th International Symposium on
  Parameterized and Exact Computation ({IPEC} 2019)}, volume 148, pages
  12:1--12:15, 2019.

\bibitem{Eto2019}
Hiroshi Eto, Tesshu Hanaka, Yasuaki Kobayashi, and Yusuke Kobayashi.
\newblock {Parameterized Algorithms for Maximum Cut with Connectivity
  Constraints}.
\newblock In {\em Proceedings of the 14th International Symposium on
  Parameterized and Exact Computation ({IPEC} 2019)}, volume 148, pages
  13:1--13:15, 2019.

\bibitem{Fomin:2008}
Fedor~V. Fomin, Dieter Kratsch, Ioan Todinca, and Yngve Villanger.
\newblock Exact algorithms for treewidth and minimum fill-in.
\newblock {\em SIAM J. Comput.}, 38(3):1058--1079, 2008.

\bibitem{Fomin:2018}
Fedor~V. Fomin, Mathieu Liedloff, Pedro Montealegre, and Ioan Todinca.
\newblock Algorithms parameterized by vertex cover and modular width, through
  potential maximal cliques.
\newblock {\em Algorithmica}, 80(4):1146--1169, 2018.

\bibitem{Fomin:2019}
Fedor~V Fomin, Daniel Lokshtanov, Saket Saurabh, and Meirav Zehavi.
\newblock {\em Kernelization: theory of parameterized preprocessing}.
\newblock Cambridge University Press, 2019.

\bibitem{Fomin:2015}
Fedor~V. Fomin, Ioan Todinca, and Yngve Villanger.
\newblock Large induced subgraphs via triangulations and {CMSO}.
\newblock {\em SIAM J. Comput.}, 44(1):54--87, 2015.

\bibitem{Gaspers:2018}
Serge Gaspers and Simon Mackenzie.
\newblock On the number of minimal separators in graphs.
\newblock {\em Journal of Graph Theory}, 87(4):653--659, 2018.

\bibitem{Hadlock:1975}
F.~Hadlock.
\newblock Finding a maximum cut of a planar graph in polynomial time.
\newblock {\em SIAM Journal on Computing}, 4(3):221--225, 1975.

\bibitem{Haglin:1991}
David~J. Haglin and Shankar~M. Venkatesan.
\newblock Approximation and intractability results for the maximum cut problem
  and its variants.
\newblock {\em IEEE Trans. Comput.}, 40(1):110--113, 1991.

\bibitem{Hanaka:2019}
Tesshu Hanaka, Hans~L. Bodlaender, Tom~C. van~der Zanden, and Hirotaka Ono.
\newblock On the maximum weight minimal separator.
\newblock {\em Theoretical Computer Science}, 796:294 -- 308, 2019.

\bibitem{Impagliazzo:2001}
Russell Impagliazzo and Ramamohan Paturi.
\newblock On the complexity of $k$-{SAT}.
\newblock {\em Journal of Computer and System Sciences}, 62(2):367 -- 375,
  2001.

\bibitem{Kloks:1998}
T~Kloks, D~Kratsch, and C.K Wong.
\newblock Minimum fill-in on circle and circular-arc graphs.
\newblock {\em Journal of Algorithms}, 28(2):272 -- 289, 1998.

\bibitem{Liu:2015}
Ching-Hao Liu, Sheung-Hung Poon, and Jin-Yong Lin.
\newblock Independent dominating set problem revisited.
\newblock {\em Theoretical Computer Science}, 562:1 -- 22, 2015.

\bibitem{Lokshtanov:2011}
Daniel Lokshtanov, D\'aniel Marx, and Saket Saurabh.
\newblock Lower bounds based on the {E}xponential {T}ime {H}ypothesis.
\newblock {\em Bull. Eur. Assoc. Theor. Comput. Sci. EATCS}, 105:41--72, 2011.

\bibitem{Milanic:2019}
Martin Milani{\v{c}} and Nevena Piva{\v{c}}.
\newblock Minimal separators in graph classes defined by small forbidden
  induced subgraphs.
\newblock In Ignasi Sau and Dimitrios~M. Thilikos, editors, {\em
  Graph-Theoretic Concepts in Computer Science}, pages 379--391, Cham, 2019.
  Springer International Publishing.

\bibitem{Skodinis:1999}
Konstantin Skodinis.
\newblock Efficient analysis of graphs with small minimal separators.
\newblock In {\em Proceedings of the 25th International Workshop on
  Graph-Theoretic Concepts in Computer Science (WG 1999)}, pages 155--166,
  Berlin, Heidelberg, 1999. Springer-Verlag.

\end{thebibliography}
\end{document}